\documentclass[12pt,reqno]{article}

\title{Complexity of Linear Subsequences of Fibonacci-Automatic Sequences}

\author{Delaram Moradi\footnote{Research supported by NSERC grant RGPIN-2024-03725.}\\
School of Computer Science\\
University of Waterloo\\
Waterloo, ON  N2L 3G1 \\
Canada\\
\href{mailto:delaram.moradi@uwaterloo.ca}{\tt delaram.moradi@uwaterloo.ca}\\
\and
Narad Rampersad\footnote{Research supported by NSERC grants RGPIN-2019-04111 and RGPIN-2025-04076.}\\
Department of Math/Stats\\
University of Winnipeg\\
515 Portage Ave.\\
Winnipeg, MB R3B 2E9\\
Canada\\
\href{mailto:n.rampersad@uwinnipeg.ca}{\tt n.rampersad@uwinnipeg.ca}\\
\and 
Jeffrey Shallit$^*$\\
School of Computer Science\\
University of Waterloo\\
Waterloo, ON  N2L 3G1 \\
Canada\\
\href{mailto:shallit@uwaterloo.ca}{\tt shallit@uwaterloo.ca}}

\usepackage[usenames]{color}
\usepackage{amsmath}
\usepackage{amssymb}
\usepackage{amsthm}
\usepackage{amsfonts}
\usepackage{amscd}
\usepackage{graphicx}

\usepackage[colorlinks=true,
linkcolor=webgreen,
filecolor=webbrown,
citecolor=webgreen]{hyperref}

\definecolor{webgreen}{rgb}{0,.5,0}
\definecolor{webbrown}{rgb}{.6,0,0}

\usepackage{color}
\usepackage{fullpage}
\usepackage{float}
\usepackage{enumitem}

\usepackage{graphics}
\usepackage{latexsym}
\usepackage{epsf}

\newcommand{\seqnum}[1]{\href{https://oeis.org/#1}{\rm \underline{#1}}}

\def\AND{\, \wedge \, }

\def\suchthat{\ : \ }

\DeclareMathOperator{\pre}{Pre}

\DeclareMathOperator{\add}{add}

\def\Enn{\mathbb{N}}

\makeatletter
\def\Ddots{\mathinner{\mkern1mu\raise\p@
\vbox{\kern7\p@\hbox{.}}\mkern2mu
\raise4\p@\hbox{.}\mkern2mu\raise7\p@\hbox{.}\mkern1mu}}
\makeatother

\begin{document}

\maketitle

\theoremstyle{plain}
\newtheorem{theorem}{Theorem}
\newtheorem{corollary}[theorem]{Corollary}
\newtheorem{lemma}[theorem]{Lemma}
\newtheorem{proposition}[theorem]{Proposition}
\newtheorem{problem}{Problem}

\theoremstyle{definition}
\newtheorem{definition}[theorem]{Definition}
\newtheorem{example}[theorem]{Example}
\newtheorem{conjecture}[theorem]{Conjecture}

\theoremstyle{remark}
\newtheorem{remark}[theorem]{Remark}

\begin{abstract}
We construct automata with input(s) in Fibonacci representation (also known as Zeckendorf representation) recognizing some basic arithmetic relations and study their number of states. 
We also consider some basic operations on Fibonacci-automatic sequences and discuss their state complexity. Furthermore, as a consequence of our results, we improve a bound in a recent paper of Bosma and Don. We also discuss the state complexity and runtime complexity of using a reasonable interpretation of B\"uchi arithmetic to actually construct some of the studied automata recognizing relations.
\end{abstract}

\section{Introduction}

In this paper, we discuss 
the state complexity of specific arithmetic relations (such as addition and multiplication) when the inputs are given in Fibonacci representation (also known as Zeckendorf representation).  We also study the state complexity of shifts and linear subsequences of Fibonacci-automatic sequences, and the computational complexity of forming the automata using an interpretation of B\"uchi arithmetic.  In a previous paper \cite{moradi_complexity_2026} we studied the same questions for $k$-automatic sequences, but the case of
Fibonacci representation presents many new and different challenges.

In Section \ref{section:background}, we provide the necessary background for the rest of the paper.
In Section \ref{section:relations}, we construct automata of $O(\log c)$ states recognizing $Y=X+c$ and $Y=X-c$ where $X, Y$ are inputs, and $c \geq 0$ is a constant. We construct an automaton of $O(n^2)$ states recognizing the relation $Y = nX+c$ where
$n \geq 1$ is a natural number, $0 \leq c < n$, and the numbers $X$ and $Y$ are inputs. In Section \ref{section:lin-subseq}, we show that if $(h(i))_{i \geq 0}$ is a fixed Fibonacci-automatic sequence generated by an $m$-state automaton, then the shifted subsequence $(h(i+c))_{i \geq 0}$ can be generated by a DFAO of $O(m^2c^2)$ states, the linear subsequence $(t(i+c))_{i \geq 0}$ of the Thue-Morse sequence 
can be generated by a DFAO of $O(c)$ states, and the linear subsequence
$(h(ni+c))_{i \geq 0}$ for $n \geq 1$ and $0 \leq c < n$ can be generated by a DFAO of 
$O(m^2n^4)$ states.
As a consequence, we improve a bound in a recent paper by Bosma and Don \cite{bosma_constructing_2024} concerning
the size of morphisms for linear subsequences of the Fibonacci word.

An example of a system using an interpretation of B\"uchi arithmetic is {\tt Walnut}, a free software package that can carry out computations with automatic sequences and evaluate statements about them phrased in first-order logic. For more information about {\tt Walnut}, see, for example,
\cite{mousavi_automatic_2021, shallit_logical_2023}.  In Section \ref{section:Buchi}, we show how the automaton generating $(h(ni+c))_{i \geq 0}$ for fixed $n \geq 1, 0\leq c < n$ is created in polynomial time using an interpretation of B\"uchi arithmetic.

In this paper, we occasionally refer to $O(\log i)$ where $i$ can be $0$ or $1$. In these cases, we adopt the usual convention that $O(\log i) = O(1)$ for $i \in \{0,1 \}$.

\section{Background} \label{section:background}

\subsection{Finite Automata}
\label{subsection:fa}

In this paper we use the familiar model of deterministic finite automaton, as discussed in \cite{hopcroft_introduction_1979}, and some variations on it.
A deterministic finite automaton (DFA)
consists of a finite set of states $Q$, an input alphabet
$\Sigma$,
an initial state $q_0$,
a set of accepting states $A$,
and a transition function $\delta:Q \times \Sigma \rightarrow Q$ that specifies the next state of the
automaton, based on the current state and the current input letter.  We extend the domain of $\delta$ to $Q \times \Sigma^*$ in the usual manner.
Acceptance is defined by whether completely processing an input causes the
DFA to enter an accepting state.
Usually the transition function $\delta$ is taken to be a total
function from $Q \times \Sigma$ to $Q$, but in this paper, we allow
it to be a partial function in order to gracefully handle dead states. 
A state $q$ of a DFA is called {\it dead} if it is not
possible to reach any accepting state from $q$ by a 
(possibly empty) path.  A minimal DFA for a language $L$
clearly has at most one dead state.  By convention, we
do not count or display dead states in this paper. A state is called {\it accessible} if it is reachable from the start state, and
{\it co-accessible} if one can reach an accepting state from it.

Another model we need is the nondeterministic finite automaton (NFA), which
is similar to the DFA, except that now the transition function
is allowed to transition to $0$ or multiple states on the same input
letter, and defines acceptance based on the existence of at least
one path that enters an accepting state after completely processing the input.

We also use the notion of deterministic finite automaton
with output (DFAO), which is like a DFA, except that outputs are associated
with states.  The output corresponding to an input is the output associated
with the last state reached.   A DFAO is called a {\it Fibonacci-DFAO\/} if its input(s) are in Fibonacci representation (see Section \ref{subsection:fib-rep}).  A Fibonacci-DFAO $M = (Q, \Sigma, \Delta, \delta, q_0, \tau)$ generates a {\it Fibonacci-automatic sequence} $(h(i))_{i \geq 0}$ in the following way: the input is an
msd-first representation $x$ of $i$ in the Fibonacci numeration system, and the output
$h(i)$ is $\tau(\delta(q_0, x))$.  We assume that such a DFAO does not reach any state on an input containing an invalid Fibonacci representation.

For a given
automatic sequence $(h(i))_{i \geq 0}$ there is a 
unique associated {\it interior sequence} $(h'(i))_{i \geq 0}$ taking
its values in $Q$, defined
by taking the minimal DFAO $(Q, \Sigma, \Delta, \delta, q_0, \tau)$ and replacing $\tau$ by the identity map on $Q$.  Here ``unique'' means up to renaming of the letters.
It follows that $(h(i))_{i \geq 0}$ is the image of
$(h'(i))_{i \geq 0}$
under the coding $q \rightarrow \tau(q)$.

The state complexity of a formal language is the number of states in the minimal DFA recognizing it and we say the state complexity of an automatic sequence is the number of states in the minimal DFAO generating it.

We will need a lemma about the subword complexity (also called ``factor complexity'') of an infinite word $\bf x$ generated by a Fibonacci-automatic sequence.
The subword complexity of $\bf x$ is a function
$\rho_{\bf x} (n)$ counting the number of distinct subwords of length $n$ appearing in $\bf x$.
\begin{lemma}
We have $\rho_{\bf x}(n) = O(n)$.
\label{lemma:subword-sc}
Let ${\bf x}$ be  Fibonacci-automatic sequence generated by an $m$-state DFAO with msd-first input.
Then $\rho_{\bf x} (n) = O(n m^2)$ for all $n \geq 1$. 
\end{lemma}

\begin{proof}
See \cite[Theorem 10.3.1]{allouche_automatic_2003}. The proof is stated there for base-$k$, but the same proof works for Fibonacci-automatic sequences.
\end{proof}

Finally, we also need the unfamiliar notion of unambiguous finite automaton with output (UFAO),
which is a blend of the notion of unambiguous
finite automaton (UFA) with that of the DFAO.  The UFAO is a nondeterministic machine where there is at most one output for every input.  The output associated with an input is the output associated with the unique accepting state reached.  The usual subset construction can be used to turn a UFAO into a DFAO, but at the risk of exponential blow-up in the number of states.

The DFA and NFA models  are described in detail in standard reference
works, such as Hopcroft and Ullman \cite[Chapter 2 and Chapter 3]{hopcroft_introduction_1979}.  For the DFAO, see, for example, \cite[Chapter 4]{allouche_automatic_2003}.
UFA's are discussed in, among many other places,
Jir{\'a}sek et al.~\cite{jozef_jirasek_jr_operations_2018}.

In this paper, we often think of automata as processing natural numbers $i$
as input, represented in some numeration system.
The input alphabet is $\Sigma_k = \{ 0,1,\ldots, k-1 \}$
for some integer $k \geq 2$, and it represents the digits
of the representation of $i
$.

This brings up a subtle point:  how do we handle leading
zeros in the input?  What if, for example, the
automaton gives a different result on input $00101$ as it does on input
$101$?  To avoid this issue, the designed automata satisfy the following criteria:
first, that 
the automata we discuss give the same behavior (accept/reject or output)
no matter how many leading zeros are appended to the front of the input.
Second, we assume that automata have a self-loop on the initial
state $q_0$ on input $0$ (or, if there are multiple inputs, on input
$[0,0,\ldots,0]$).  These conventions are crucial when we are working
with multiple integer inputs of different lengths.

\subsection{Fibonacci Representation} \label{subsection:fib-rep}

We let $\Enn = \{0,1,2,\ldots \}$.
Define the Fibonacci numbers, as usual, by
$F_0 = 0$, $F_1 = 1$, and $F_n = F_{n-1} + F_{n-2}$ for $n \geq 0$.

Fibonacci representation was introduced by
Lekkerkerker \cite{lekkerkerker_voorstelling_1951} and
Zeckendorf \cite{zeckendorf_representation_1972}, although it
can be found 
in an earlier, much more general form in a paper of Ostrowski \cite{ostrowski_bemerkungen_1922}.
The Fibonacci representation of $n \in \Enn$ writes $n$ 
as a sum of distinct Fibonacci numbers $F_i$ for $i \geq 2$; that is,
$n = \sum_{2 \leq i \leq t} e_i F_{i}$ for $e_i \in \{0,1\}$.
To ensure uniqueness of the representation, no two
consecutive Fibonacci numbers can be used; that is, in a valid representation we have
$e_i  \cdot e_{i+1} = 0$ for $1 \leq i \leq t$.  A representation can
be associated with the binary word of its coefficients
$e_t  \cdots e_2$;
note that representations are written with the most significant bit
at the left.  We need a way to convert an arbitrary
binary word $x = e_t \cdots e_2$ to the integer it
represents. For this we write $[x] = \sum_{1 \leq i \leq t} e_i F_{i}$.
For example, we have $[01001] = 6$.
The canonical Fibonacci representation of an integer $i$ is the word $x$, without leading zeros, such that $[x] = i$.  For this we use
the notation $(i)$.
We say a word $z$ is a {\it valid right extension\/} of a word $x$
if $x$ is a prefix of $z$ and $z$ contains no occurrence of $11$.

Sometimes we will need representations of pairs of natural
numbers.  To do so, we use the alphabet $\Sigma_2 \times \Sigma_2$.
This may require padding the shorter representation with leading zeros.
Thus, for example, the word
$[0,1][0,0][1,1][0,0][1,0]$ represents the pair $[4,11]$.
If $w$ and $x$ are words of the same length over the
alphabet $\Sigma$, then by $w \times x$ we mean
the element $y \in (\Sigma^2)^*$ such that
the first component is $w$ and the second is $x$. These notions are similarly extended to represent tuples of more than two natural numbers.

For the rest of this paper, we assume that the inputs to automata are in Fibonacci representation and given most significant digit first format, unless otherwise stated.

The Fibonacci word ${\bf f} = f_0 f_1 f_2 \cdots = 01001010\cdots$ is an infinite binary word that can be defined in a number of equivalent ways \cite{berstel_fibonacci_1986}. The simplest for us is that $f_i$ is the last bit of the
Fibonacci representation of $i$. It is also defined as the fixed point of the morphism $0 \rightarrow 01$, $1 \rightarrow 0$. It can be generated by a DFAO of $2$ states, as illustrated in
Figure~\ref{fig:fib-word-automaton}.

\begin{figure}
    \centering
    \includegraphics[width=0.5\linewidth]{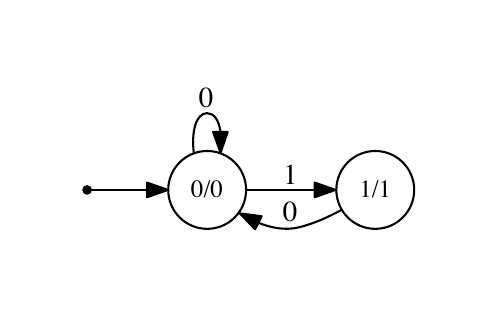}
    \caption{Automaton for the Fibonacci word $\bf f$.}
    \label{fig:fib-word-automaton}
\end{figure}

We now state some facts about Fibonacci representation.  The first is easy and the second is mentioned by Reble \cite{reble_zeckendorf_2008}, and the third is a consequence of the Fibonacci representation being given by the greedy algorithm.   Let $\alpha = (1+\sqrt{5})/2$ and $\beta = (1-\sqrt{5})/2$, the zeros of the polynomial $X^2-X+1$.
\begin{lemma}
\leavevmode
\begin{itemize}
    \item[(a)] Let $x \in \Sigma_2^*$.  Then
    $[x00] = [x]+[x0]$.
    \item[(b)] Let $x$ be a valid Fibonacci representation.
    Then $-\beta^2 < [x0]-\alpha[x] < -\beta$.
    \item[(c)] Let $x,y$ be valid Fibonacci representations.
    Then $x < y$ in lexicographic order iff $[x] <[y]$.
\end{itemize}
\label{lemma:[x0]}
\end{lemma}

Next we state a lemma that we will use in
Section~\ref{section:lin-subseq}.
\begin{lemma}
Suppose $x$ is a word, $a \in \{0,1\}$ and $xa$ is a valid Fibonacci representation.
Then $[x] = \lfloor ([xa]+2)/\alpha \rfloor - 1$.
\label{lemma:[x]-floor}
\end{lemma}

\begin{proof}
\leavevmode

\noindent Case 1: $a = 0$.  Then by Lemma~\ref{lemma:[x0]} (b),
we have
$$ \alpha[x]-\beta^2 < [xa] < \alpha[x]-\beta.$$
Add $2$ and divide by $\alpha$ to get
$$[x] + 1 < {{[xa]+2} \over \alpha} < [x] + \alpha < [x]+2.$$
The result now follows.

\medskip

\noindent Case 2: $a = 1$.  Then, since $xa$ is a valid Fibonacci
representation, we must have $x = x'0$.  Following the
argument in Reble \cite{reble_zeckendorf_2008}, let $x' = e_t \cdots e_2$
so that 
$[x'] = \sum_{2\leq i \leq t} e_i F_i$.  Then
\begin{align*}
    [x'00] &= \sum_{2 \leq i \leq t} e_i F_{i+2}\\
    &= \sum_{2 \leq i \leq t} e_i (\alpha F_{i+1} + \beta^{i+1}) \\
    &= \alpha \sum_{2 \leq i \leq t} e_i F_{i+1} + 
        \sum_{2 \leq i \leq t} \beta^{i+1} \\
        &= \alpha [x'0] + \sum_{2 \leq i \leq t} \beta^{i+1}.
\end{align*}
Hence 
$$-\beta^2 = \beta^3 + \beta^5 + \beta^7 + \cdots <
[x'00] - \alpha [x'0] < \beta^4 + \beta^6 + \beta^8 + \cdots = \sqrt{5} - 2.$$
Now, adding $3 + \alpha [x'0]$ to both sides, we get
$\alpha[x'0] +\alpha^2 < [x'01] + 2  < \alpha[x'0] + \sqrt{5} +1$, or
in other words,
$\alpha[x] + \alpha^2 < [xa]+2 < \alpha[x] + \sqrt{5} + 1 $.
Dividing through by $\alpha$ gives
$[x] + 1< [x] +\alpha  < {{[xa]+2}\over \alpha} < [x] + 2$,
which proves the result for this case.
\end{proof}

We also need the following lemma. 
\begin{lemma}
Suppose $x, y$ are two words of the same length,
and $a, a', b, b' \in \Sigma_2$ and $D(x, y) = [y]-n[x]$ for some $n\geq 1$.  Then
$$ D(xa'a,yb'b) = D(x,y) + D(xa',yb') + b-na + b'-na'.$$
\label{lemma:fib-update-rule}
\end{lemma}
\begin{proof}
Lemma~\ref{lemma:[x0]} (a) gives us that $[x00]= [x]+[x0]$.  Thus
$[x a' a] = [x] + [x0] + 2a' + a$ and
$[y b' b] = [y] + [y0] + 2b' + b$.
Hence
\begin{align*}
D(xa'a, yb'b) &= [y b' b]-n[x a' a]  \\
&= [y] + [y0] + 2b' + b - n([x] + [x0] + 2a' + a) \\
&= ([y]-n[x]) + ([yb'] - n[xa']) +  b'-na' + b-na \\
&= D(x,y) + D(xa', yb') + b'-na' + b-na.
\end{align*}
\end{proof}

\section{Recognizing Relations} \label{section:relations}

In this section we study the state complexity of recognizing some simple relations; for example the state complexity of the automaton accepting $x \times y$ if and only if $[x]+c=[y]$.

\subsection{Addition}

\begin{theorem} \label{theorem:[x]+c=[y]}
    For all integers $c \geq 0$, there exists an msd-first Fibonacci automaton with at most $O(\log  c)$
states accepting $x \times y$ where $[x] + c = [y]$.
\end{theorem}

\begin{proof}
We assume $c\geq 1$, since the case $c = 0$ is trivial.
 We design a Fibonacci automaton $M_c$ that reads input $x\times y$ in parallel and accepts if and only if $[y]=[x]+c$.
We modify the same approach we used previously for the case of base-$k$ representation \cite{moradi_complexity_2026}.
Namely, we first find a range of values $I$ containing $c$ such that if $[y]-[x]$ ever falls outside $I$, then $[y']-[x']$ stays outside $I$ for all right extensions $x'$ of $x$ and $y'$ of $y$. 

Suppose $[y]<[x]$. By Lemma~\ref{lemma:[x0]} (c) we see that $y < x$ under the lexicographic order.   Hence, if $y'$ is a right-extension of $y$ and $x'$ is a right-extension of $x$, we have
$y' < x'$ under the lexicographic order, and so $[y'] < [x']$.  

Now suppose $[y] > [x]+ c $.  We have
\begin{align*}
    [yb] - [xa] &= [y0] - [x0] + b -a 
    > (\alpha[y]-\beta^2) - (\alpha[x]-\beta) - 1 \\
    &= \alpha ([y]-[x]) - \beta^2 + \beta - 1 
    \geq \alpha (c+1) - \beta^2 + \beta - 1 
    = \alpha c- \beta^2 
    > c.
\end{align*}
Therefore, we can take the range $I$ to be $[0,c]$.

Our automaton $M_c=(Q, \Sigma_2, q_0, \delta, A)$ is constructed as follows.
\begin{align*}
     Q &\subseteq \{[a, b, d', d] \suchthat a, b \in \{ 0,1 \} \text{ and } d, d' \in I\}, \\
     q_0 &= [0, 0, 0, 0], \\
     \delta([a', b', d'', d'], [a, b]) &= [a, b, d', d] \ \text{where} \ a,b,a',b \in \{0,1\}, aa', bb' \not=11, \text{ and } \\
     & \quad\quad d= d'
     + d'' +b' - a' + b - a, \\
     A &= \{[a, b, d', d] \suchthat a, b \in \{0,1 \} \text{ and } d = c\}.
 \end{align*}
The correctness of the transition rule now follows from Lemma~\ref{lemma:fib-update-rule} (with $n=1$) and an easy induction on the length of the input.

Our next task is to define $Q$ more precisely. 
Define $d' = [y]-[x]$, $d = [yb]-[xa]$,
$\varepsilon_x = [x0] - \alpha [x]$, and $\varepsilon_y = [y0] - \alpha [y]$.    Thus
\begin{align*}
d &= [yb] - [xa] 
= [y0]-[x0] + b-a \\
&= \alpha[y] + \varepsilon_y - \alpha[x] - \varepsilon_x + b-a = \alpha d' + \varepsilon_y - \varepsilon_x + b-a.
\end{align*}
Rewriting, we get
\begin{equation}
 d' - d/\alpha = {{\varepsilon_x -  \varepsilon_y +a-b} \over \alpha}.
 \label{equation:d'-d/alpha}
 \end{equation}
There are now three cases to consider.  By Lemma~\ref{lemma:[x0]} (b) we know $-\beta^2 < \varepsilon_x, \varepsilon_y < -\beta$. Hence the cases are as follows.

\medskip

\noindent{\it Case 1}:  $a=b$.
Then $d/\alpha -1 < d'  < d/\alpha + 1$.

\medskip

\noindent{\it Case 2}:  $a = 0$, $b = 1$.   Then $d/\alpha -2 < d'  < d/\alpha$.

\medskip

\noindent{\it Case 3}:  $a = 1$, $b = 0$.  Then $d/\alpha < d'  < d/\alpha +2$.

\medskip

This gives eight possible forms of states corresponding to a particular $d$:
\begin{align*}
   & \{ [0,0,\lfloor d/\alpha \rfloor,d],
[0,0, \lceil d/\alpha \rceil,d],
[1,1,\lfloor d/\alpha \rfloor,d],
[1,1,\lceil d/\alpha \rceil,d],\\
& \quad [0,1,\lfloor d/\alpha \rfloor, d],
[0,1, \lfloor d/\alpha \rfloor -1, d],
[1,0, \lceil d/\alpha \rceil, d],
[1,0, \lceil d/\alpha \rceil + 1, d] \}.
\end{align*}

Starting from a final state, of the form $[a,b,d',c]$, we determine the possible predecessors of this state, and the possible predecessors of those, etc.  Note that by Equation (\ref{equation:d'-d/alpha}) if $[y]-[x] \geq 4$, then for any right extension of $x, y$ the difference only increases.  We group the predecessors together in levels.  The states of the form $[a,b,d',c]$ are at level $0$, and the predecessors of these are at level $1$, and so forth.
We now have to bound the number of states in each level, and the number of levels. 

States in level $i$ are of the form $[a,b,d',d]$ and
$d \in [D_i, D_i + r_i]$ for certain $D_i, r_i$.  Note that $D_0 = c$ and $r_0 = 0$.  Since the smallest possible $d'$ in a level $i$ state is $\lfloor D_i/\alpha \rfloor - 1$ and the largest possible is $\lceil (D_i +r_i)/\alpha \rceil + 1$, it follows that
$r_{i+1} < r_i/\alpha + 4$.  Starting with $r_0 = 0$, using this inequality four times shows that $r_i \leq 8$ for all $i$. Thus there are at most $8+1 = 9$ possible values of $d$ among all states of any given level.

To estimate the number of possible levels, note that Equation \eqref{equation:d'-d/alpha} gives $d' < (d+2)/\alpha$ and thus
$D_{i+1} + r_{i+1} < (D_i+10)/\alpha < D_i/1.1$
provided $D_i \geq 22$.  Thus
there are at most $\log_{1.1} c = O(\log c)$ levels until $D_i \leq 22$.
When the largest $d$ in a level is smaller than $22$, we cease our computation of levels and simply include all states $[a,b,d',d]$
with $d \leq 22$ (a constant number).

This shows that there are only $O(\log c)$ co-accessible states and we are done.
\end{proof}
A similar proof can be provided for $[x]-c = [y]$.

\begin{theorem}
    For all $c \geq 0$, there exists an automaton with at most $O(\log  c)$
    states accepting $x \times y$ where $[x] - c = [y]$.
\end{theorem}

\begin{proof}
    Consider the automaton provided for the proof of Theorem \ref{theorem:[x]+c=[y]}. Here the only difference is that we keep track of $[x]-[y]$ instead of $[y]-[x]$. So we only need to change the transition function to the following:
    \[
    \delta([a',b',d,d'],[a,b]) =
    [a,b,g,d],  \text{ if $g \in I$ where $g = d + d' + a-b + a'-b'$}\text{ and $aa', bb' \not= 11$}.
    \]
\end{proof}

\subsection{Multiplication}

Let $n \geq 1$ be a natural number and $0 \leq c < n$.
We provide an upper bound on the number of states required for an automaton to 
recognize the relation $Y = nX + c$.

We let $L_{n,c}$ be the language
of all words over $\Sigma_2 \times \Sigma_2$ where the first component represents some integer
$[x]$ and the second component represents $n[x] +c$.
More precisely, 
$$ L_{n,c} := \{ x \times y \in (\Sigma_2 \times \Sigma_2)^* \suchthat [y] = n[x] + c\} .$$
For example, Figure~\ref{figure:M_2,0}
depicts a minimal DFA recognizing $L_{2,0}$ computed by {\tt Walnut}.  Here the state
numbered $0$ is the initial state (denoted by the headless arrow) and
the accepting states $0$, $3$, $6$, and $8$ are indicated by double circles.
\begin{figure}[htb]
\begin{center}
\includegraphics[width=6.5in]{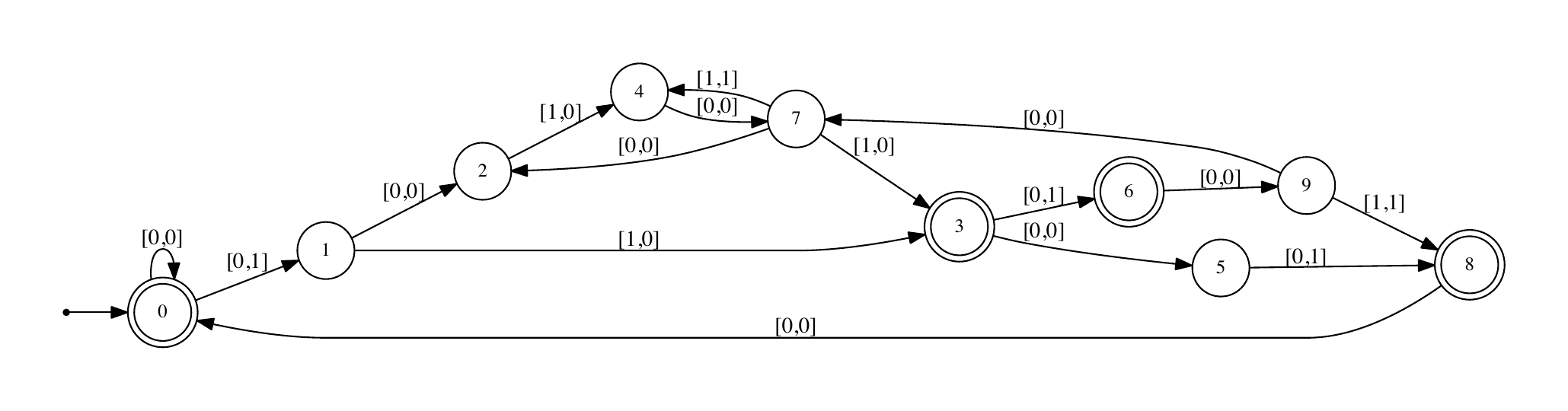}
\end{center}
\caption{Multiplication by $2$ in Fibonacci representation.}
\label{figure:M_2,0}
\end{figure}

\begin{theorem} \label{theorem:L_n,c-sc}
There is a DFA  of $O(n^2)$ states that accepts $L_{n,c}$.
\end{theorem}

A result related to Theorem \ref{theorem:L_n,c-sc} was proved by Charlier et al.~\cite{charlier_minimal_2011}.  They proved that $2n^2$ states are necessary and sufficient
for a DFA
to accept multiples of the constant $n$ with input in the Fibonacci numeration system.
However, despite the close relationship of their result and ours, we see
no direct way to get our result from theirs.

For input in Fibonacci representation, we use the same strategy outlined for recognizing this relation with input in base $k$ from \cite{moradi_complexity_2026}.  Namely, we use an automaton $M_{n,c}$ to keep track
of the difference 
function $D(x,y) := [y] - n [x]$. Here there are complications that arise from four things:
\begin{itemize}
\item In base $k$, reading a single letter $a$
from the input changes the integer read from $n$ to
$kn+a$.  But the formula is more involved for
Fibonacci representation.

\item It is no longer sufficient to maintain just a single
$d$, the current difference. We also need to maintain another related quantity,
$d'$, the previous difference.  

\item We have to ensure that the inputs are
valid Fibonacci representations.  This means we have
to store, in the state, the last letter read for both
input words.

\item It is no longer so simple to determine the interval $I_n$
for $d$ and $d'$, outside of which no right extension would be accepted.
\end{itemize}

As we mentioned,
the idea is that on input $xa' \times yb'$, the automaton should
keep track of $d = D (xa', yb')$.  However, considering Lemma \ref{lemma:fib-update-rule}, this would not be sufficient; we also need to keep track of $d' = D(x,y)$.
Moreover, we need to determine the interval $I_n$ so that if
$d$ ever falls outside this interval, we know that no right extension
of the input can ever lie in $L_n$.  Thus states are $4$-tuples
of the form $[a',b', d, d'] \in \Sigma_2 \times \Sigma_2 \times I_n \times I_n$.

Assuming the existence of $I_n$, the definition of the transition function now follows immediately from Lemma~\ref{lemma:fib-update-rule}.
Namely, 
$M_{n,c} = (Q, \Sigma_2, q_0, \delta, A)$ is defined as follows:
\begin{align*}
Q &= \{ [a',b',d,d'] \suchthat d, d' \in I_n \text{ and } a',b' \in \Sigma_2 \}, \\
\Sigma &=  \Sigma_2 \times \Sigma_2, \\
q_0 &= [0, 0, 0, 0], \\
\delta([a',b',d,d'],[a,b]) 
&=
    [a,b,g,d],  \text{ if $g \in I_n$ where $g = d + d' + b-na + b'-na'$}\\
    & \quad\quad \text{ and $aa', bb' \not= 11$};
     \\
A &= \{ [a', b', d, d'] \suchthat d = c \}.
\end{align*}

So task that remains is  to determine the
interval $I_n$. We show that we can take $I_n$ to be the closed interval $[-n, 2n-1]$.
\begin{lemma}
Suppose $D(x,y) \leq -n-1$.  Then $D(xa,yb) \leq -n-1$.
\label{lemma:I_n-lower}
\end{lemma}
\begin{proof}
We have
\begin{align*}
D(xa,yb) &= [yb] - n [xa] \\
&= [y0]-n[x0] + b-na \\
&< (\alpha [y] - \beta)  - n (\alpha[x] - \beta^2) + 1 \\
&= \alpha D(x,y) - \beta + n \beta^2 + 1\\
&\leq \alpha (-n-1) - \beta + n \beta^2 + 1 \\
& = -n(\alpha - \beta^2) + 1-\alpha-\beta \\
& = -n(\alpha - \beta^2) \\
& < -n \\ 
& \leq -n-1.
\end{align*}
\end{proof}

\begin{lemma}
Let $a, b \in \Sigma_2$, and
suppose $D(x,y) \geq 2n$.  Then $D(xa,yb) \geq n$.
\label{lemma:I_n-upper-2n}
\end{lemma}
\begin{proof}
We have
\begin{align*}
D(xa,yb) &= [yb]-n[xa] \\
&= ([y0] + b) - n([x0] + a) \\
&\geq [y0] - n[x0] - n \\
&> (\alpha [y] - \beta^2) - n(\alpha[x] - \beta) -n \\
&= \alpha([y]-n[x]) -\beta^2 + n\beta - n \\
&= \alpha D(x,y) - \beta^2 + n\beta - n \\
&\geq 2n\alpha - \beta^2 + n \beta - n \\
&= (2\alpha + \beta - 1)n - \beta^2 \\
&> n - 1.
\end{align*}
\end{proof}

\begin{lemma}
Let $a,b, a', b' \in \Sigma_2$ such that
either $a = 0$ or $a' = 0$.
Suppose $D(x,y) \geq n$ and $D(xa',yb') \geq n$.  Then
$D(xa'a, yb'b) \geq n$.
\label{lemma:I_n-upper-n}
\end{lemma}

\begin{proof}
From Lemma~\ref{lemma:fib-update-rule} we have
\begin{align*}
D(xa'a, yb'b) &= D(x,y) + D(xa',yb') + b-na + b'-na' \\
&\geq 2n -n(a + a') \\
& \geq n.
\end{align*}
Note that we know
that at least one of $a, a'$ equals $0$, for otherwise,
we do not have a valid Fibonacci representation.
\end{proof}

\begin{theorem} \label{theorem:I_n}
We can take $I_n = [-n,2n-1]$.
\end{theorem}
\begin{proof}
Lemma~\ref{lemma:I_n-lower} shows that if $D(x,y) \leq -n-1$,
then the same inequality holds for every right extension
of $x$ and $y$.  

Lemmas~\ref{lemma:I_n-upper-2n} and \ref{lemma:I_n-upper-n} together show
that if $D(x,y) \geq 2n$, then $D(x',y') \geq n$
holds for every right extension $x'$ of $x$ and
$y'$ of $y$.
\end{proof}

Now a simple induction on
$|x| = |y|$ shows that 
$$\delta(q_0, xa' \times yb') = 
[a', b', D(xa',yb'), D(x,y)],$$
provided that
no prefix $x' \times y'$ of $xa' \times yb'$ has $D(x', y') \not\in I_n$.

We can now prove Theorem~\ref{theorem:L_n,c-sc}.
\begin{proof}[Proof of Theorem \ref{theorem:L_n,c-sc}]
The automaton $M_{n,c}$ we have constructed consists of
states of the form
$[a',b',d, d']$ where $d, d' \in I_n$ and
$a', b' \in \Sigma_2$.  This gives a total
of $2\cdot2\cdot3n \cdot 3n$ states,
which is $O(n^2)$.
\end{proof}

\begin{remark}
We point out that the DFA $M_{n,c}$ in our construction is not, in general,
minimal.  However, we can construct our automaton in quadratic time and then apply a well-known minimization algorithm for
DFAs by Hopcroft \cite{hopcroft_introduction_1979} or Valmari \cite{valmari_fast_2012} to get the minimal automaton
recognizing $Y=nX+c$ in $O(n^2 \log n)$ time.

The number of states in the minimal DFA recognizing the language
$L_{n,0}$ forms sequence \seqnum{A372846} in the
On-Line Encyclopedia of Integer Sequences (OEIS).   The values
for $1 \leq n \leq 10$ (not counting the dead state), as
computed by the free software {\tt Walnut}, are
$2, 10, 23, 40, 59, 85, 114, 146, 181, 224 $.
Numerical evidence suggests that this quantity is roughly
$2n^2 + \alpha^2 n$, but we do not know how to prove this.  A better understanding of the reachable states of $M_{n,c}$ might help reduce the bound.
\end{remark}

\section{Linear Subsequences}
\label{section:lin-subseq}

A sequence $(h(i))_{i \geq 0}$ over a finite alphabet is said to be {\it Fibonacci-automatic} if it is generated 
by a DFAO $M_h$ reading an integer input $i$ in Fibonacci representation,
and returning the value of $h(i)$ as the output associated with the last
state reached.  More precisely, if $M_h = (Q_h, \Sigma_2, \Delta_h, q_{0,h}, \delta_h, \tau_h)$, then $h(i) = \tau_h(\delta_h(q_{0,h}, x))$ for all
 valid representations $x \in \Sigma_2^*$ with $[x] = i$. We assume that there are no transitions corresponding to
two consecutive inputs of $1$; in other words, if there is a state $q$ reachable on $1$ from some other state, then $\delta_h(q,1)$ is undefined. For more information about such sequences, see \cite{mousavi_decision_2016}.

\begin{theorem} \label{theorem:sc-h(i+c)}
    Let $(h(i))_{i \geq 0}$ be a Fibonacci-automatic sequence generated by a DFAO of $m$ states. For a constant $c \geq 0$, there is a DFAO of $O(m^2c^2)$ states generating $(h(i+c))_{i \geq 0}$.
\end{theorem}

\begin{proof}
    \sloppy
    Let $(h'(i))_{i \geq 0}$ be the interior sequence of $(h(i))_{i \geq 0}$ generated by $M=(Q, \Sigma_2, q_0, \delta, \Delta, \tau)$.
    We use a method similar to the one used for base-$k$ in \cite{moradi_complexity_2026} and create an automaton $M'$ that on input $x$ reaches a state of the form $
    [[h'([x]), a'_0], \ldots, [h'([x]+c), a'_c]]
    $ where $a'_i$ is the last letter in $([x]+i)$ and after reading the letter $a$ transitions to the next state 
    $[[h'([xa]), a_0], \ldots, [h'([xa]+c), a_c]]$.
    
    The automaton $M'=(Q', \Sigma_2, q'_0, \delta', \Delta, \tau')$ is constructed as follows.
    The states in $Q'$ are of the form $[[p_0, a_0], \ldots, [p_c, a_c]]$ where $p_i \in Q$ and $a_i \in \Sigma_2$.
    The initial state $q'_0$ is $[[h'(0),0], \ldots, [h'(c), a_c]]$ where $a_c$ is the last letter of $(c)$.
    Furthermore, we have $\tau'([[p_0, a_0], \ldots, [p_c, a_c]]) = \tau(p_c)$.
    
    The transition function $\delta'$ is more complicated and requires explanation.
    If the input to $\delta'$ are $[[p_0, a_0], \ldots, [p_c, a_c]]$ and $a$, it computes the output as follows. First, it creates a list. Then it goes through each $[p_i, a_i]$ in order; if $a_i=0$, then $[\delta(p_i, 0), 0]$ and $[\delta(p_i, 1), 1]$ are added to the list; otherwise only $[\delta(p_i, 0), 0]$ is added to the list. Finally, if $a=0$, the function outputs the $0$-th to $c$-th element; otherwise it outputs the $1$-st to $(c+1)$-th element. 
    
    Now let us find the number of states in $M'$. Since $p_0, \ldots, p_c$ form a length-$(c+1)$ subsequence of $(h'(i))_{i \geq 0}$, there are $O(m^2c)$ possibilities. Furthermore, the $a_0, \ldots, a_c$ form a length-$(c+1)$ subsequence of the infinite Fibonacci word and there are $O(c)$ possibilities for them.
    So the automaton $M'$ for $(h(i+c))_{i \geq 0}$ has $O(m^2c^2)$ states.
\end{proof}

\begin{theorem}
Let $(t(i))_{i \geq 0}$ be the Fibonacci-Thue-Morse sequence; in other words, $t(i)$ is the number of $1$s in the Fibonacci representation of $i$ modulo $2$. For a constant $c \geq 0$, the number of states in the minimal DFAO generating $(t(i+c))_{i \geq 0}$ is $O(c)$.
\end{theorem}

\begin{proof}
    Let us create the automaton $M_c$ for $(t(i+c))_{i \geq 0}$ based on the construction provided in Theorem \ref{theorem:sc-h(i+c)}.
    On input $x$, the automaton $M_c$ reaches a state of the form $[ [t'([x]), a'_0], \ldots,  [t'([x]+c), a'_c]$ where $a'_j$ is the last letter of $([x]+j)$.
    Consider the automaton generating $(t(i))_{i \geq 0}$ from  \cite{shallit_subword_2021}. We can see that $t'([x]+j)$ determines $a'_j$. So from Lemma \ref{lemma:subword-sc}, we conclude that there are $O(c)$ states in $M_c$. 
\end{proof}

We suspect that the number of states in a minimal automaton generating $(t(i+c))_{i\geq 0}$ has a linear lower bound too. The following is left as an open problem.

\begin{problem}
    Prove the number of states in a minimal automaton generating $(t(i+c))_{i\geq 0}$ is $\Theta(c)$. 
\end{problem}

In a forthcoming paper \cite{moradi_state_2026}, we study the state complexity of generating $(f(i+c))_{i\geq 0}$ where $(f(i))_{i \geq 0}$ is the the Fibonacci word {\bf f} and we prove it is $O(\log c)$.

\begin{theorem}
Let $n \geq 1$ and $0 \leq c < n$, and let $(h(i))_{i \geq 0}$ be a Fibonacci-automatic sequence generated by a DFAO of $m$ states.  There is a DFAO of
$O(m^2n^4)$ states recognizing the linear subsequence
$(h(ni+c))_{i \geq 0}$.
\label{theorem:lin-subseq-sc}
\end{theorem}

Let $(h(n))_{n \geq 0}$ be an arbitrary Fibonacci-automatic sequence computed by
a DFAO 
$$M_h = (Q_h, \Sigma_2, \Delta, \delta_h, q_{0,h}, \tau_h),$$
where $\Delta$ is a finite alphabet and $\tau_h$ is the output function
associated with $M_h$.  Fix $n$ and $c$. We now form a
UFAO $M_u = (Q_u, \Sigma_2, \Delta, \delta_u, q_{0,u}, A_u, \tau_u )$ as follows:
\begin{align*}
Q &= \{ [a',b',d,d',q'] \suchthat a', b' \in \Sigma_2
\text{ and } d, d' \in I_n \text{ and } q' \in Q_h \}, \\
\delta([a',b',d,d',q'], a) &=
\{ [a,b,g, d,q] \suchthat b \in \Sigma_2,\
 g = d + d' + b'-na' + b-na,\ d \in I_n,\\
 &\quad\quad\quad \text{$aa', bb' \not= 11$, and } q = \delta_h(q', b) \}, \\
 q_{0,u} &= \{ [0,b',0,i,p] \suchthat 0 \leq i \leq 2n-1,\ 
    b '= {\bf f}[i], \ p = \delta_h(q_{0,h}, (i) ) \}, \\
    A_u &= \{ [a',b',d,d',q'] \suchthat d=c \}, \\
\tau_u ([a',b',d, d', q']) &= \tau_h(q').
\end{align*}
Conceptually, $M_u$ is formed by taking a Cartesian product of $M_{n,c}$
with $M_h$ that takes an input $x$, 
``guesses'' $y$, and simulates $M_{n,c}$
on $x \times y$ and $M_h$ on $y$.  The output of the UFAO is the unique output
of $M_h$ on $y$ corresponding to $d = c$.   

The initial state $q_{0,u}$ requires some
explanation.   It corresponds to all inputs $x$ such that
$[x]=0$ (in other words, $x = 0^j$ for some $j$) and all possible $y$ such that $[y]-n[x] \in I_n$.   This is required because
we want the DFAO we construct to return the correct result no matter how many leading zeros are prefixed to the input.

Next, we use the subset construction on $M_u$, obtaining a DFAO $M$.
On input $x$ this DFAO computes $h(n[x] + c)$.  
The states of $M$ are subsets of states of $M_u$.  We first prove a
lemma that characterizes the subsets of states that can appear in
a state of $M$.

\begin{lemma}
For each state $p$ in $M$
there exist $a \in \Sigma_2$ and integers $S, T$ such 
that $[S, S+T] \subseteq [-n,2n-1]$ and every state of $p$ is of the form
$[a, b_i, d_i, d'_i, s_i]$ 
with $0 \leq i \leq T$ and $d_i = S+i$.
\label{lemma:lin-subseq-consec}
\end{lemma}
In other words, the third components of the states of $p$ are distinct,
and consist of all integers in the range $[S, S+T]$ exactly once.
\begin{proof}
Consider a state $p$ reachable via the input $xa$.   The subset construction
generates those states of $M_u$ reachable via the input $xa \times yb$ where
$r = |y| = |x|$.  The corresponding integers $[yb]$ are therefore precisely those in the interval $[0,F_{r+2}-1]$.  The corresponding $d_i$ are of the form $[yb]-n[xa]$.  However, not all of these result in valid states
because the $d'_i, d_i$ could fall outside the interval $[-n, 2n-1]$.
Thus the range of permissible $d_i$ is an intersection of intervals, and hence
itself is a subinterval of $[-n,2n-1]$.
\end{proof}

\begin{proof}[Proof of Theorem \ref{theorem:lin-subseq-sc}]
Consider an input $xa$ to $M$.
By Lemma~\ref{lemma:lin-subseq-consec} we know that
on this input $M$ reaches a set of
states corresponding to inputs $xa\times y_i b_i$ in $M_u$, where $d_i = [y_i b_i]-n[xa]=  S+i$  for
$0 \leq i \leq T$, some word $y_i$ and letter $b_i$.  Let $[y_0 b_0] = Y$.  Then $[y_i b_i] = Y+i$
for $0 \leq i \leq T$.  Noting that
$d'_i = [y_i] - n[x]$, 
define $e_i = d'_{i+1} - d'_i$ for $0 \leq i < T$. 
Let $g$ be the function that
maps $t$ to the integer obtained by dropping the last
bit in the Fibonacci representation of $t$; Lemma~\ref{lemma:[x]-floor}
shows that $g(t) = \lfloor (t+2)/\alpha \rfloor - 1$.
Then $d'_i = g(Y+i) -n[x]$.
The triples $(b_i, e_i, s_i )$ are given by $({\bf f}[Y+i], g(Y+i+1)-g(Y+i), h'(Y+i))$
for $0 \leq i \leq T$, and this formula completely specifies {\it all\/} the $b_i$ and $d'_i$ and $s_i$,
once $d'_0$ is chosen.  Namely,
$d'_i = d'_0 + \sum_{0 \leq j < i} e_j$.

On input $x$ the automaton generating $(g(j+1)-g(j))_{j \geq 0}$ only needs to keep track of the last two letters read; if any of them is $1$ then $g([x]) = 1$, otherwise $g([x]) =0$. 
Now think of the triples $(b_i, e_i, s_i)$  as a subsequence of the infinite sequence of triples 
${\bf x} := ({\bf f}[j], g(j+1)-g(j), h'(j))_{j\geq 0}$, which 
is an automatic sequence because $({\bf f}[j])_{j \geq 0}$, $(g(j+1)-g(j))_{j \geq 0}$, and the interior sequence $(h'(j))_{j\geq 0}$ are automatic sequences.
Furthermore, the sequence ${\bf x}$ is generated by a DFAO of $O(m)$ states.
The finite sequence
$({\bf f}[Y+i], g(Y+i+1)-g(Y+i), h'(Y+i))_{0 \leq i \leq T}$ is then completely specified by saying which factor of length $T+1$ of $\bf x$ it is, and by
Lemma~\ref{lemma:subword-sc}, there are $O(m^2n)$ of them.

Summarizing, there are
\begin{itemize}
\item[(a)] $O(n)$ choices for $d_0 \in [-n, 2n-1]$,
\item[(b)] $O(n)$ choices for $d'_0 \in [-n, 2n-1]$,
\item[(c)] $O(n)$ choices for $T$,
\item[(d)] $O(m^2n)$ choices for which particular subword of the automatic sequence $\bf x$ gives $(b_i, d'_{i+1}-d'_i, s_i)$ for $0 \leq i \leq T$,
\item[(e)] $O(1)$ choices for $a$.
\end{itemize}

Therefore there are $O(m^2n^4)$ possible states.
\end{proof}

\begin{example}
    Let us find the minimal DFAO for ${\bf f}[2i]$.  We can use the
    free software {\tt Walnut} for this, using the commands
    \begin{verbatim}
    def f2 "?msd_fib F[2*i]=@1":
    combine F2 f2=1:
    \end{verbatim}
    This gives us the DFAO in Figure~\ref{figure:f[2i]}.
    \begin{figure}[htb]
    \begin{center}
    \includegraphics[width=5in]{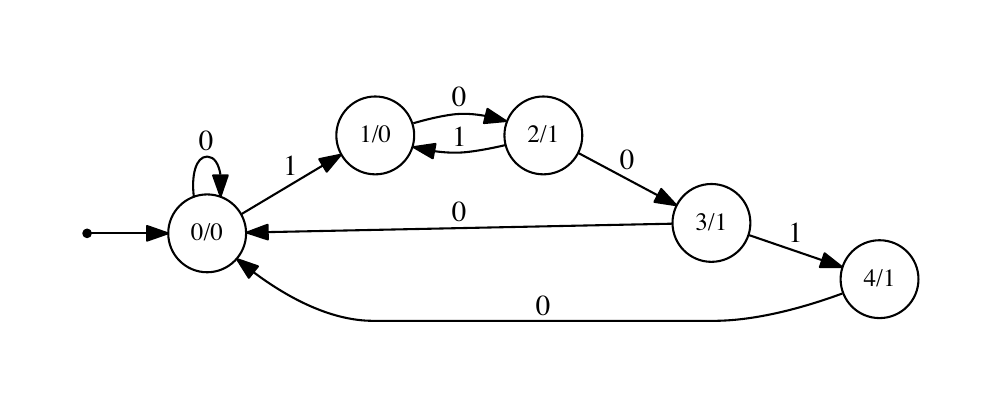}
    \end{center}
    \caption{A Fibonacci DFAO for ${\bf f}[2i]$.}
    \label{figure:f[2i]}
    \end{figure}

    Note that {\tt Walnut}'s output is a  minimized automaton. So Figure \ref{figure:f[2i]}
    is the minimized version of the automaton built by the construction provided in this paper.
     
\end{example}

\begin{remark}
We point out that the automaton we have constructed for
$(h(ni+c))_{i \geq 0}$ is not minimal.  However, we can find
the minimal automaton in polynomial time by a minimization algorithm. The minimization algorithm running in $O(n \log n)$ time by Hopcroft \cite{hopcroft_n_1971} can be adapted for DFAOs as described by van Spaendonck \cite{van_spaendonck_automatic_2020}.

For the special case where $c = 0$ and $h(i) = {\bf f}[i]$, the infinite Fibonacci word, the smallest DFAO for the linear subsequence
${\bf f}[ni]$ has a number of states given by OEIS sequence
\seqnum{A385021}.  The first 10 terms, corresponding to
$1 \leq n \leq 10$, are $2, 5, 10, 17, 27, 36, 52, 65, 78, 103$.
\end{remark}

\subsection{Application to a Problem of Bosma and Don}

As is well-known, a $k$-automatic sequence can also be generated as the image (under a coding) of the fixed point  $\gamma^\omega (0)$ of a prolongable morphism $\gamma$ that maps each letter of some alphabet $\Sigma$ to a word of $\Sigma^*$ of length $k$ \cite{allouche_automatic_2003}.   Here
``prolongable'' means that the image of the letter $0$ begins with $0$.
In fact, the appropriate morphism can be deduced immediately from the automaton, as follows:  the letters of the alphabet are the states $q$ of the automaton; the value $\gamma(q)$ is given by the states reached
on inputs $0, 1, \ldots , k-1$.  The coding is defined by the
output associated with each state.
Roughly the same thing is true for Fibonacci-automatic sequences, except now the values of 
$\gamma(q)$ are given by the outputs on
$0$ and $1$ (if the latter is possible by the
rules of the numeration system).
\begin{example}
 From the automaton in Figure~\ref{figure:f[2i]} we can read off the following representation for the subsequence $({\bf f}[2i])_{i \geq 0}$ in
    terms of a morphism $\gamma$ and coding $\rho$:
    \begin{align*}
    \gamma(0) &= 01,  & \quad \tau(0) &= 0, \\
    \gamma(1) &= 2, & \quad \tau(1) &= 0, \\
    \gamma(2) &= 31, & \quad \tau(2) &= 1, \\
    \gamma(3) &= 04, & \quad \tau(3) &= 1, \\
    \gamma(4) &= 0, & \quad \tau(4) &= 1 .
    \end{align*}
When we iterate the morphism $\gamma$ we get the sequence
$01231040\cdots$ and after applying the coding $\tau$ we get
$00110000\cdots$.
\end{example}

Bosma and Don \cite{bosma_constructing_2024} studied the  linear subsequences of the Fibonacci word {\bf f}.   They were mostly concerned with the size of the smallest morphism generating the sequence $({\bf f}[ni+c])_{i \geq 0}$,
and they were able to prove an exponential upper bound on the size of this morphism.  Here the size of the morphism is the total number of symbols
required to write it down, that is,
$\sum_{a \in \Sigma}  |\gamma(a)| $.

As a consequence of our results, we can reduce the size of the needed morphism from exponential to polynomial.

\begin{corollary}
Let $n \geq 1$ and $0 \leq c < n$, and let
$(h(i))_{i \geq 0}$ be a Fibonacci-automatic sequence.
There is a morphism of size
$O(n^4)$ generating the subsequence $(h(ni+c))_{i \geq 0}$.
\end{corollary}

\begin{proof}
As we saw in Theorem~\ref{theorem:lin-subseq-sc}, the resulting
automaton has $O(m^2n^4)$ states. Hence the corresponding morphism
has size $O(m^2n^4)$.
\end{proof}

\section{Construction by B\"uchi Arithmetic} \label{section:Buchi}

Similar to what had been previously discussed for the case of base-$k$ \cite{moradi_complexity_2026}, given an automatic sequence, we can compute the automaton for a linear subsequence using the free software {\tt Walnut} 
\cite{mousavi_automatic_2021,shallit_logical_2023}.  {\tt Walnut} uses an interpretation of B\"uchi arithmetic (the logical theory of the natural numbers together with addition and  the function $V_k(n)$, the largest power of $k$ dividing $n$, for some fixed $k \geq 2$) and implements a 
decision procedure due to B\"uchi \cite{buchi_weak_1960,bruyere_logic_1994} and takes a first-order logical statement and translates it into the
corresponding automaton. The users of such software are interested in knowing how long it takes to algorithmically create the corresponding automata; in this section we study this time complexity. We use the total number of transitions of an automaton $M$ as a proxy for the computational complexity of constructing $M$. 

In particular, in its current version, the algorithms for creating the automata we discussed in previous sections are not used.  Instead, more general techniques, suitable for any addable numeration system, are used.  This means that the size of intermediate automata may differ from the minimal automata we constructed, and the running time can similarly be larger. Sometimes the complexity bounds in this section include the addition of $O(\log i)$ or $O(i^j)$ components for some constant $i$; these addition components are not significant compared to multiplication in $O(\log i)$ or $O(i^j)$ components for some constant $i$.

Furthermore, the minimization steps in the implementation for an $n$-state automaton can use an algorithm running in $O(n \log n)$ time  by Hopcroft \cite{hopcroft_n_1971} or Valmari \cite{valmari_fast_2012}. The algorithm by Hopcroft can be slightly modified to be applied to DFAOs as described by van Spaendonck \cite{van_spaendonck_automatic_2020}. Note that in an $n$-state DFA, the number of transitions is $O(n)$.  
The minimization is applied to each DFA or DFAO created. In this paper, we sometimes only apply partial minimizations and ignore the potential effects of fully minimizing automata in order to be able to analyze the runtime complexity and obtain an upper bound; the upper bound on state complexity and runtime complexity will remain correct regardless and the runtime complexity obtained reflects the computational cost of a complete minimization.

Consider the automaton $M_{\add}$ from Mousavi et al.~\cite{mousavi_decision_2016} that accepts input $x \times y \times z$ if $[x]+[y]=[z]$. 
On input $x \times y \times z$ this automaton leads to the state denoted by the sequence $([x0^n] + [y0^n]-[z0^n])_{n \geq 0}$. We modify this automaton to keep track of some extra information. Consider the automaton $M'_{\add}$ created from $M_{\add}$ that on input $xa \times yb \times ze$ leads to the state denoted by $[s, s', e]$ where $s=([xa0^n] + [yb0^n]-[ze0^n])_{n \geq 0}$ and $s' = ([x0^n] + [y0^n]-[z0^n])_{n \geq 0}$. 

If we have $M_{\add} = (Q, \Sigma^3_2, \delta, q_0, A)$, we create $M'_{\add}=(Q', \Sigma^3_2, \delta', q'_0, A')$ as follows.
\begin{align*}
    Q' &= Q \times Q \times \Sigma_2, \\
    \delta'([s, s', e], [a, b, e]) &= [\delta(s, [a, b, e]), s , e], \\
    q'_0 &= [q_0, q_0, 0], \\
    A' &= \{[s, s', e] \suchthat s \in A\}.
\end{align*}

We use $M'_{\add}(x, y, z)[i]$ to show the $i$-th component of the state $M'_{\add}$ reaches on input $x \times y \times z$ for $i \in \{0, 1, 2\}$; for example, $M'_{\add}(x, y, z)[0]$ is $([x0^n] + [y0^n]-[z0^n])_{n \geq 0}$. 
The automaton $M'_{\add}$ has $O(1)$ states.

\begin{lemma} \label{lemma:=c-sc}
The minimal automaton for recognizing the relation
$[x]_k = c$ has $O(\log c)$ states. 
\end{lemma}

\begin{proof}
Such an automaton consists of a linear chain of nodes labeled with the Fibonacci representation of $c$, which has
length $O(\log c)$.  
\end{proof}

\begin{theorem} \label{theorem:=c}
    Let $c \geq 0$ be a fixed constant. The automaton recognizing the relation $[x] = c$ can be computed in $O((\log ^2 c)(\log \log c))$ time.
\end{theorem}
\begin{proof}
    We recursively create the automaton $M_{=c}$.

    The base case is $c=0$ and the automaton has $1$ state.
    
    If $c$ is even, we recursively compute a minimal DFA $M_{=(c/2)}$ recognizing $[y] = c/2$ on input $y$ and then use the first-order expression
    \[
    \exists y \ M_{=(c/2)}(y) \land M_{\add}(y, y, x)
    \]
    which is translated into an automaton.  We know the automaton for $M_{\add}$ has $O(1)$ states and by Lemma \ref{lemma:=c-sc}, the automaton $M_{=(c/2)}$ has $O(\log c)$ states. First we apply the product construction to $M_{=(c/2)}$ and $M_{\add}$ such that the $M_{=(c/2)}$ part processes $y$ and the $M_{\add}$ part processes $y \times y \times x$; the resulting automaton has $O(\log c)$ states. Each state in the product automaton is of the form $[q, p]$ where on input $x \times y$, $q$ is a state from the minimal automaton $M_{=(c/2)}$ representing the number of letters from $(c/2)$ matched with $y$ by Lemma \ref{lemma:=c-sc} and $p$ is a state from $M_{\add}$ such that  $p(0)= 2[y]-[x]$.
    Then we apply the $\exists y$ quantifier by removing the component corresponding to $y$ from the transitions and we get an NFA. In the NFA there are no two same transitions from a single state leading to different states. So after subset construction the resulting DFA is the same as the NFA. Then we do minimization and the resulting DFA is the automaton from Lemma \ref{lemma:=c-sc} with $O(\log c)$ states.

    If $c$ is odd, we recursively compute a DFA $M_{=(c-1)}$ recognizing $[y] = c-1$ on input $y$ and then use the first-order expression
    \[
    \exists y \ M_{(c-1)}(y) \land M_{\add}(y, 1, x)
    \]
    which is translated into an automaton. Following a construction similar to the one described for the case where $c$ is even, we get the DFA from Lemma \ref{lemma:=c-sc} for constant $c$ with $O(\log c)$ states.

    Taking minimization steps into account, we get the following recursive formula for the time complexity:
    \[  T(c) =
\begin{cases}
     O((\log c) (\log \log c)) + T(c/2), & \text{if $c \bmod 2 = 0$}; \\
    O((\log c) (\log \log c))  + T(c-1), & \text{if $c \bmod 2 = 1$}.
\end{cases}
\]
So the total time is $ O((\log^2  c) (\log \log c))$.
\end{proof}

\begin{theorem} \label{theorem:[x]+c=[z]-time}
    Let $c \geq 0$ be a fixed constant. The automaton recognizing the relation $[x] + c =[z]$ can be computed in $O((\log^2 c)(\log \log c))$ time.
\end{theorem}
\begin{proof}
Assume we have modified $M'_{\add}$ again so that on input $xa \times yb \times ze$, it keep track of $a$ in addition to $e$ and its states are of the form $[s, s', a, e]$.  We can use the following expression to create the automaton $M_c$ recognizing $[x] + c=[z]$ on input $x \times z$. Moreover, we have automaton $M_{=c}$ recognizing $[y]=c$ from Theorem \ref{theorem:=c}.

    \[
    \exists y \ M'_{\add}(x, y, z) \AND M_{=c}(y).
    \]
    First, we apply the product construction to $M'_{\add}$ and $M_{=c}$ such that the $M'_{\add}$ part processes $x\times y \times z$ and the $M_{=c}$ part processes $y$. The resulting automaton $M'_c$ has $O(\log c)$ states.
     Let $\pre(c, i)$ denote the first $i$ letters in $(c)$. On input $xa\times yb \times ze$ where $yb$ is $0^*y'b$ and $y'b$ is a prefix of $(c)$ the automaton $M'_c$ leads to a state $[q, p]$ where $q$ is from $M'_{\add}$ and $p=|y'b|$ is from $M_{=c}$ indicating how many letters from $y$ have been matched with $(c)$. Therefore, we have $q=[s, s', a, e]$ where $s(0)=[xa] + [yb] - [ze]$, $s'(0)=[x]+[y]-[z]$, and $[\pre(c, p)] = [yb]$. So $[ze] - [xa] = [\pre(c, p)] - s(0)$ and $[z]-[x]=[\pre(c, p-1)]-s'(0)$.  So to each state $[q, p]$ in $M'_c$ a tuple $[d, d', a, e]$ can be attributed where $d = [ze]-[xa]$ and $d'= [z]-[x]$ and two states reachable by the same $xa, ze$ cannot have different $[d, d', a, e]$ attributed to them.

    Then we remove the $y$ component from transitions and determinize the automaton to obtain a DFA. Consider a (non-dead) state reachable on input $xa \times ze$ in the DFA. This state consists of some $[q, p]$ such that $[ze] - [xa] = [\pre(c, p)] - s(0)$ and $[z]-[x]=[\pre(c, p-1)]-s'(0)$. In the NFA there are no two same transitions from a single state leading to different states. So after subset construction the resulting DFA is the same as the NFA.

From the proof of Theorem \ref{theorem:[x]+c=[y]}, we know the automaton recognizing $[x]+c= [z]$ on input $x \times z$ only needs to keep track of a difference $[z] - [x]$ in range $[0, c]$ with $O(\log c)$ states. So after some minimization, the automaton $M_c=(Q, \Sigma_2, q_0, \delta, A)$ can be formally defined as follows.  Let $I_c = [0, c]$. 
 \begin{align*}
     Q &= \{[a, b, d, d'] \suchthat d, d' \in I_c\}, \\
     q_0 &= [0, 0, 0, 0], \\
     \delta([a', b', d, d'], [a, b]) &= [a, b, g, d] \ \text{where} \ g= d + d' +b' - a' + b - a, \\
     A &= \{[a, b, d, d'] \suchthat d = c\}.
 \end{align*}
 Note that not all states in the formal definition above are necessarily reachable.

    \sloppy
Considering the time spent on creating $M_{=c}$ and minimization steps, it takes $O((\log^2 c)(\log \log c))$ time to create the automaton $M_c$ recognizing $[x] + c = [z]$ on input $x \times z$.   
\end{proof}

We continue with the automaton recognizing the relation $[z] = n[x]$ where $n$ is a fixed constant. 
There are a few possible recursions for implementing this automaton.
\begin{itemize} [nosep]
    \item $\exists y \ M_{n-1}(x,y) \AND M_{\add}(x, y, z)$.
    \item  $\exists y_1, y_2 \ M_{\lfloor n/2 \rfloor}(x,y_1) \AND M_{\lceil n/2 \rceil}(x,y_2) \AND M_{\add}(y_1, y_2, z)$.
\end{itemize}
However, a third recursive implementation that follows is more efficient than the two above.

If $n$ is even, we
recursively compute a DFA $M_{n/2}(x,y)$ recognizing $[y] =  (n/2)[x]$
and then use the first-order expression
$$ \exists y \ M_{n/2}(x,y) \AND M'_{\add}(y, y, z),$$
which is translated into an automaton by a direct product construction
for the $\AND$, and projection of the $2$-nd coordinate to remove
transitions on $y$. 

If $n$ is odd, we recursively compute
a DFA $M_{n-1}(x,y)$ recognizing $[y]_k = (n-1)[x]_k$ and use the expression
$$ \exists y \ M_{n-1}(x,y) \AND M_{\add}'(x, y, z).$$
In both cases, the translation could conceivably generate a nondeterministic automaton (by the projection) and determinizing it could take exponential time.
However,  we show that this is not the case.

\begin{theorem} \label{theorem:M_n-buchi}
        Let $n \geq 1$ be a fixed constant.
    An automaton for recognizing $[y] = n[x]$ on input $x \times y$ can be computed in $O(n^6 \log^2 n)$ time. 
\end{theorem}

\begin{proof}
    
We first prove by induction on $n$ that we can recursively create the automaton $M_n$ with states of the form $[a, b, d, d']$ such that on input $xa \times yb$ the automaton reaches a state where $d = [yb]-n[xa]$ and $d'=[y]-[x]$, and $d, d' \in I_n = [-n, 2n-1]$.

Base case: $n=1$.

In this case, a simple $1$-state automaton recognizes $[y]=[x]$. The single state corresponds to $[a, b, d, d'] = [0, 0, 0, 0]$.

For the induction step, assume the result is true for all $n'<n$; we prove it for $n$.
There are two cases.

\bigskip\noindent{\bf
Case 1:} $n \bmod 2 = 0$.

By the induction hypothesis, we have constructed a DFA $M_{n/2}$ with $O(n^2)$ states corresponding to all possible $[a, b, d, d']$ where $a, b \in \Sigma_2$ and $d, d' \in I_{n/2} = [-(n/2), 2(n/2)-1]$.
In this case, we use the formula
$$ \exists y \ M_{n/2}(x,y) \AND M'_{\add}(y, y, z).$$
First we create an automaton $M'_{n}$ that is the product construction of $M_{n/2}$ and $M'_{\add}$. So this automaton has $O(n^2)$ states. Consider an input $xa \times yb \times ze$ for $M'_{n}$ leading to some state $[q, p]$ where $q$ is from $M_{n/2}$ and $p$ is from $M'_{\add}$. By the construction of $M'_{n}$, we know $q=[a, b, d, d']$ where $d=[yb]-(n/2)[xa]$, $d'=[y]-(n/2)[x]$ and  $p=[s, s', e]$ where $s=(2[yb0^n]-[ze0^n])_{n \geq 0}$, $s'=(2[y0^n]-[z0^n])_{n \geq 0}$.
So we have $2d-s(0) = [ze]-n[xa]$ and $2d'-s'(0) = [z]-n[x]$.
Therefore, for each state reachable by $xa \times yb \times ze$ in $M'_n$ a tuple $[a, e, \bar{d}, \bar{d'}]$ can be attributed where $\bar{d}=[ze]-n[xa]$ and $\bar{d'}=[z]-n[x]$. Furthermore, two input with the same $xa, ze$ cannot lead to states with different $[a, e, \bar{d}, \bar{d'}]$ attributed to them. 

Next, we project away the $y$ component.  The resulting automaton is an NFA, and we use the subset construction to create a DFA.  Consider a state from the DFA reachable on $xa \times ze$. This state consists of some pairs $[q, p]$ from the NFA such that $q=[a, b, d, d']$, $p=[s, s', e]$ and
we can create $[a, e, 2d-s(0), 2d'-s'(0)]$ from $[q, p]$ where $[ze]-n[xa]=2d-s(0)$ and $[z]-n[x]=2d'-s'(0)$.

In each state of the DFA, using an argument similar to the one used in the proof of Theorem \ref{theorem:lin-subseq-sc}, there are $O(n^4)$ possibilities for the $q=[a, b, d, d']$ parts. 
All $[q, p]$ in this state should have the same $2d-s(0)$ and $2d'-s'(0)$ ($O(n^2)$ possibilities).
If we look at the transition table from Hamoon et al.~\cite{mousavi_decision_2016} and consider all $\delta_{\add}(s', [b, b, e])=s$ where $\delta_{\add}$ is the transition function of $M_{\add}$, then for each $b$ the $(s'(0), s(0))$ pairs obtained from transitions are unique. So considering we have already determined the $b, d, d'$ of each $[q, p]$ and we have fixed the $2d-s(0)$ and $2d'-s'(0)$, the $s, s'$ are now determined as well ($O(1)$ possibilities). So there are $O(n^6)$ states in the DFA.

From the proof of Theorem \ref{theorem:L_n,c-sc}, we know the automaton recognizing $n[x]= [z]$ on input $x \times z$ only needs to keep track of a difference $[z] - n[x]$ in range $I_n = [-n, 2n-1]$. So after some minimization, the DFA  $M_{n}=(Q, \Sigma^2_k, \delta, q_0, A)$ can be formally defined as follows.
\begin{align*}
Q &= \{ [a, e, d, d']  \suchthat a, e \in \Sigma_2, \ d, d \in I_n] \}, \\
q_0 &= [0, 0, 0, 0], \\
\delta([a', e', d, d'], [a, e]) &= [a, e, g, d] \ \text{where} \ g = d +d' + e'-na' +e -na,
 \\
A &= \{[a, e, d, d'] \suchthat d=0\}.
\end{align*}

\bigskip\noindent{\bf 
Case 2:} $n \bmod 2 = 1$.

By the induction hypothesis, we have a DFA $M_{n-1}$ with $O(n^2)$ states corresponding to all possible $[a, b, d, d']$ where $a, b \in \Sigma_2$ and $d, d' \in I_{n-1} = [-(n-1), 2(n-1)-1]$.
In this case, we use the formula
$$ \exists y \ M_{n-1}(x,y) \AND M'_{\add}(x, y, z)$$
to create an automaton $M'_{n}$ from the usual product construction for $M_{n-1}$ and $M'_{\add}$. So this automaton has $O(n^2) $ states. Consider an input $xa \times yb \times ze$ for $M'_{n}$ leading to some state $[q, p]$ where $q$ is from $M_{n-1}$ and $p$ is from $M'_{\add}$.
By the construction of $M'_n$, we know $q=[a, b, d, d']$ where $d=[yb]-(n-1)[xa]$, $d'=[y]-(n-1)[x]$ and $p=[s, s', e]$ where $s=([xa0^n] + [yb0^n]-[ze0^n])_{n \geq 0}$, $s'=([x0^n] + [y0^n]-[z0^n])_{n \geq 0}$.
So we have $d-s(0)=[ze]-n[xa]$ and $d'-s'(0)=[z]-n[x]$.
Therefore, for each state reachable by $xa \times yb \times ze$ in $M'_n$ a tuple $[a, e, \bar{d}, \bar{d'}]$ can be attributed where $\bar{d}=[ze]-n[xa]$ and $\bar{d'}=[z]-n[x]$. Furthermore, two input with the same $xa, ze$ cannot lead to states with different $[a, e, \bar{d}, \bar{d'}]$ attributed to them.

Next, we project away the $y$ component.  The resulting automaton is an NFA, and we use the subset construction to create a DFA.  Consider a state from the DFA reachable on $xa \times ze$. This state consists of some pairs $[q, p]$ from the NFA such that $q=[a, b, d, d']$, $p=[s, s', e]$ and
we can create $[a, e, d-s(0), d'-s'(0)]$ from $[q, p]$ where $[ze]-n[xa]=d-s(0)$ and $[z]-n[x]=d'-s'(0)$.

In each state of the DFA, using an argument similar to the one used in the proof of Theorem \ref{theorem:lin-subseq-sc}, there are $O(n^4)$ possibilities for the $q=[a, b, d, d']$ parts and there are $O(1)$ possibilities for $e$ in $p$.
All $[q, p]$ in this state should have the same $d-s(0)$ and $d'-s'(0)$ ($O(n^2)$ possibilities).
If we look at the transition table from Hamoon et al.~\cite{mousavi_decision_2016} and consider all $\delta_{\add}(s', [a, b, e])=s$ where $\delta_{\add}$ is the transition function of $M_{\add}$, then for each $a, b, e$ the $(s'(0), s(0))$ pairs obtained from transitions are unique. So considering we have already determined the $a, b, e, d, d'$ of each $[q, p]$ and we have fixed the $d-s(0)$ and $d'-s'(0)$, the $s, s'$ are now determined as well ($O(1)$ possibilities). So there are $O(n^6)$ states in the DFA.

From the proof of Theorem \ref{theorem:L_n,c-sc}, we know the automaton recognizing $n[x]= [z]$ on input $x \times z$ only needs to keep track of a difference $[z] - n[x]$ in range $I_n = [-n, 2n-1]$. So after some minimization, the DFA  $M_{n}=(Q, \Sigma^2_k, \delta, q_0, A)$ can be formally defined as follows.
\begin{align*}
Q &= \{ [a, e, d, d']  \suchthat a, e \in \Sigma_2, \ d, d \in I_n] \}, \\
q_0 &= [0, 0, 0, 0], \\
\delta([a', e', d, d'], [a, e]) &= [a, e, g, d] \ \text{where} \ g = d +d' + e'-na' +e -na,
 \\
A &= \{[a, e, d, d'] \suchthat d=0\}.
\end{align*}

Next we analyze the time complexity of creating the above automaton. We use $T(n)$ to show the time required to recursively create $M_n$. We use the number of transitions traversed as a measure of the time required.

In Case 1, the number of transitions traversed to obtain $M_n$ is: $O(1)$ transitions for $M'_{\add}$, $O(n^2)$ transitions for the product construction of $M_{n/2}$ and $M'_{\add}$, $O(n^2)$ transitions for projecting away the $y$ component, $O(n^6)$ transitions for the determinized automaton, in total $O(n^6)$ transitions in addition to the number of transitions traversed for creating $M_{n/2}$ and the minimizations.

In Case 2, the number of transitions traversed to obtain $M_n$ is: $O(n^2)$ transitions for $M_{n-1}$, $O(1)$ transitions for $M'_{\add}$, $O(n^2)$ transitions for the product construction of $M_{n-1}$ and $M'_{\add}$, $O(n^2)$ for projecting away the $y$ component, $O(n^6)$ for the determinized automaton, in total $O(n^6)$ transitions in addition to the number of transitions traversed for creating $M_{n-1}$ and the minimizations.

Putting everything together, we get the following recursive formula:
\[
\begin{cases}
    T(n) = O(n^6\log n) + T(n/2), \\
    T(n) = O(n^6 \log n) + T(n-1).
\end{cases}
\]
So we have $T(n) = O(n^6\log^2 n)$.
\end{proof}

\begin{theorem} \label{theorem:n[x]+c=[z]-time}
     Let $n \geq 1$, $0 \leq c < n$ be
     fixed constants. The automaton recognizing $n[x] + c = [z]$ on input $x \times z$ can be computed in $O(n^6 \log^2 n)$ time.
\end{theorem}

\begin{proof}
Let $M_c$ be the automaton from Theorem \ref{theorem:[x]+c=[z]-time} and let $M_n$ be the automaton from Theorem \ref{theorem:M_n-buchi}. We say $M_c(y, z)$ is $1$ if and only if $[y] + c = [z]$. So we can use the following expression to create $M_{n, c}$:
\[
\exists y \  M_n(x, y) \ \land  M_c(y, z).
\]

We first apply product construction to $M_n$ and $M_c$. The resulting automaton $M'_{n,c}$ has $O(n^2\log c)$ states.

Consider an input $xa \times yb \times ze$ for $M'_{n, c}$ leading to some state $[q, p]$ where $q$ is from $M_n$ and $p$ is from $M_c$.
By the construction of $M'_{n, c}$, we know that $q=[a, b, d_q, d_q']$ where $d_q=[yb]-n[xa]$, $d_q'=[y]-n[x]$ and $p=[b, e, d_p, d_p']$ where $d_p = [ze]-[yb]$, $d_p'=[z]-[y]$. So we have $d_q + d_p = [ze]-n[xa]$ and $d'_q + d'_p = [z]-n[x]$.
Therefore, to each state in the automaton $M'_{n, c}$ a tuple $[a, e, d, d']$ can be attributed where $d=[ze]-n[xa]$ and $d'=[z]-n[x]$ and no two states reachable by the same $xa, ze$ can have different $[a, e, d, d']$ attributed to them.

Next, we project away the $y$ component from $M'_{n, c}$.  The resulting NFA needs to be determinized by subset construction to get a DFA.
Consider a state from the DFA reachable on $xa \times ze$. This state consists of some $[q, p]$ from the NFA such that $q=[a, b, d_q, d_q']$, $p=[b, e, d_p, d_p']$ and $d_q + d_p = [ze]-n[xa]$, $d'_q + d'_p = [z]-n[x]$. 
In each state of the DFA, using an argument similar to the one used in the proof of Theorem \ref{theorem:lin-subseq-sc}, there are $O(n^4)$ possibilities for the $q=[a, b, d_q, d_q']$ parts. Then we fix $e, d, d'$ for this state of DFA with $O(n^2)$ possibilities. Since $d_q + d_p =d$ and $d'_q + d'_p = d'$, we can now determine $d_p, d'_p$ as well. So the DFA has $O(n^6)$ states.

From the proof of Theorem \ref{theorem:L_n,c-sc}, we know the automaton recognizing $n[x]+c= [z]$ on input $x \times z$ only needs to keep track of a difference $[z] - n[x]$ in range $I_n = [-n, 2n-1]$. So after some minimization, the DFA $M_{n, c}=(Q, \Sigma^2_k, \delta, q_0, A)$ can be formally defined as follows.
\begin{align*}
Q &= \{ [a, e, d, d']  \suchthat a, e \in \Sigma_2, \ d, d' \in I_n \}, \\
q_0 &= [0, 0, 0, 0], \\
\delta([a', e', d, d'], [a, e]) &= [a, e, g, d] \ \text{where} \ g =d + d' + e'-na' + e-na,
 \\
A &= \{[a, e, d, d'] \suchthat d=c\}.
\end{align*}
Taking into account the minimization steps and the time spent on creating $M_n$, the time spent on creating $M_{n,c }$ is $O(n^6 \log^2 n)$.
\end{proof}

\begin{theorem} \label{theorem:h(ni+c)-time}
    Let $n \geq 1$, $0 \leq c < n$ be fixed constants. Given an $m$-state  DFAO generating $(h(i))_{i \geq 0}$, we can compute the DFAO for $(h(ni+c))_{i \geq 0}$  in $O( n^6 \log^2 n + m^2n^4 \log (m^2n^4) )$ time.
\end{theorem}

\begin{proof}
\sloppy
We are given a DFAO $M_h=(Q_h, \Sigma_k, \delta_h, q_{0,h}, \Delta_h, \tau_h)$ for $(h(i))_{i \geq 0}$ with $m$ states and we want to create an automaton for $(h(ni+c))_{i \geq0}$. We show the output of $M_h$ on $y$ by $M_h(y)$ and we have $M_h(y) = h([y])$.
Furthermore, we have $M_{n, c}$ from above such that $M_{n,c}(x, y)$ is $1$ if and only if $n[x] + c = [y]$.

We first use the product construction for $M_{n, c}$ on input $x \times y$ and $M_h$ on input $y$ to get the automaton $M'$. Consider a state $[q, p]$ reachable on $xa\times yb$ in $M'$ where $q$ is from $M_{n, c}$ and $p$ is from $M_h$. We have $q = [a, b, d, d']$ where $d=[yb]-n[xa]$, $d'=[y]-n[x]$ and $p=\delta(q_{0,h}, yb)$. This automaton has $O(mn^2)$ states.

Next, we project the $y$ component away and use the subset construction to obtain a DFAO from the non-deterministic automaton. On input $xa$ the automaton reaches a state consisting of some $[q, p]$ where $q=[a, b, d, d']$ such that there is $yb$ where $d=[yb]-n[xa]$, $d'=[y]-n[x]$,  $d, d' \in I_n$ and $p = \delta(q_{0, h}, yb) \in Q_h$. 
In each state of the DFAO, using an argument similar to the one used in the proof of Theorem \ref{theorem:lin-subseq-sc}, there are $O(m^2n^4)$ possibilities. Therefore, the DFAO has $O(m^2n^4)$ states.

Taking the time required for constructing $M_{n, c}$ and minimization steps into account, creating $M$ given $M_{n, c}$ takes $O( n^6 \log^2 n + m^2n^4\log (m^2n^4))$ time.
\end{proof}

\section{Conclusion}
In a previous paper \cite{moradi_complexity_2026}, we studied topics similar to this paper, but with base-$k$ input instead of input in Fibonacci representation. We can attempt to study these topics in other numeration systems and solve similar problems. For example, we state the following as an open problem.

\begin{problem}
    What is the state complexity of $(r(i+c))_{i \geq 0}$ where $r$ is the Tribonacci word?
\end{problem}

We suspect our results generalize to any generalized automatic sequence, defined over a Pisot numeration system of degree $d$.  In this case recognizing the relation $Y = nX+ c$ can be done using $O(n^d)$ states, and computing the DFAO for a linear subsequence can be done in $O(n^{d+2})$ states; studying these is left as a potential future research topic.

\end{document}